\documentclass[a4paper]{article}

\usepackage{tikz}
\usetikzlibrary{arrows,automata,shapes,positioning}

\usepackage{times,authblk}
\usepackage{enumerate}
\usepackage{MnSymbol}
\usepackage{paralist}
\usepackage{amsmath,amsthm,amsfonts,bbm}
\usepackage{comment}

\usepackage{todonotes,fullpage}
\presetkeys{todonotes}{inline}{}

\newtheorem{theorem}{Theorem}
\newtheorem{lemma}[theorem]{Lemma}
 
\theoremstyle{definition}

\newenvironment{Property}[1]{\textbf{#1.} }{}

 \newcommand{\LARE}{LARE}
\newcommand{\pra}{PRA}
\newcommand{\opra}{O\pra{}}
\newcommand{\piszczalka}[3]{#1 \to^{#2} #3}
\newcommand{\q}[1]{{\color{black}{\small{\texttt{#1}}}}}

\newcommand{\labelling}[1]{\lambda_{\textrm{#1}}}

\newcommand{\qam}{QAM}

\newcommand{\aut}{\mathcal{A}}
\newcommand{\padding}{\square}
\newcommand{\Path}{\mathrm{route}}

\newcommand{\fmin}{\textsc{Min}}
\newcommand{\fmax}{\textsc{Max}}
\newcommand{\fcount}{\textsc{Count}}
\newcommand{\fsum}{\textsc{Sum}}

\newcommand{\op}{\sim}
\newcommand{\definedBy}[2]{\mathrel{{#1}{:=}{#2}}}
\newcommand{\lang}{\mathcal{L}}
\newcommand{\Tvalue}[1]{\eta^G(#1)}

\newcommand{\Z}{\mathbb{Z}}
\newcommand{\F}{{\cal F}}
\newcommand{\M}{M}
\newcommand{\z}{@}

\newcommand{\set}[1]{\{#1\}}
\newcommand{\naww}[1]{\llangle #1 \rrangle}
\newcommand{\query}[1]{\lsem #1 \rsem}
\newcommand{\tuple}[1]{{\langle #1 \rangle}}
\newcommand{\Paragraph}[1]{\noindent \textbf{#1}.}

\newcommand{\NP}{\textsc{NP}}
\newcommand{\NL}{\textsc{NL}}

\newcommand{\PSpace}{\textsc{PSpace}}

\title{Querying Best Paths in Graph Databases\thanks{
A conference version fo this paper has been accepted to FSTTCS 2017. 
This work has been supported by Polish National Science Center grant UMO-2014/15/D/ST6/00719.}} 
\author[1]{Jakub Michaliszyn}
\author[1]{Jan Otop}
\author[1]{Piotr Wieczorek}

\affil[1]{Institute of Computer Science, University of Wrocław \\
  \texttt{\{jmi,jotop,piotrek\}@cs.uni.wroc.pl}}

\begin{document}
\maketitle

\begin{abstract}
Querying graph databases has recently received much attention. We propose a new approach to this problem, which balances competing goals of 
expressive power, 
language clarity 
and
computational complexity. 
A distinctive feature of our approach is the ability to express properties of minimal (e.g. shortest) and maximal (e.g. most valuable) paths satisfying given criteria.
To express complex properties in a modular way, we introduce  labelling-generating ontologies.
The resulting formalism is computationally attractive -- queries can be answered in non-deterministic logarithmic space in the size of the database.
\end{abstract}

\section{Introduction}\label{sec:intro}
Graphs are one of the most natural representations of data in a number of important applications such as modelling transport networks, social networks, technological networks (see surveys \cite{Angles+16,Wood12,Baeza13}). The main strength of graph representations is the possibility to naturally represent not only the data itself, but also the links among data. 
Effective search and analysis of graphs is an important factor in reasoning performed in various AI tasks.  
This motivates the study of query formalisms for graph databases, which are capable of expressing properties of paths.

One of the most challenging problems of recent years is to process big data, typically too large to be stored in the modern computers' memory.
This stimulates a strong interest in algorithms working in logarithmic space w.r.t. the size of the database (data complexity)~\cite{Calvanese+06a,Artale+07a,Barcelo+12}.
At the same time, even checking whether there is a path between two given nodes is already \NL{}-complete, so \NL{} is the best complexity for any expressive graph query language.

\Paragraph{Our contribution} 
We propose a new approach to writing queries for graph databases, in which labelling-generating ontologies are first-class citizens.
It can be integrated with many existing query formalisms. However, in order to make the presentation clear we introduce the concept by defining   
a new language \opra{}. \opra{} features \NL{}-data complexity, good expressive power and a modular structure.
The expressive power of \opra{} strictly subsumes the expressive power of popular existing formalisms with same complexity (see Fig. \ref{fig:diagram}). Distinctive properties expressible in \opra{} are based on aggregation of data values along paths and computation of extremal values among aggregated data. One example of such a property is ``$p$ is a path from $s$ to $t$ that has both the minimal weight and the minimal length among all paths from $s$ to $t$''.

To ease the presentation, we define \opra{} in two steps. First, we define the language \pra{}, whose main components are three types of constraints: P\emph{ath}, R\emph{egular} and A\emph{rithmetical}. 
We use path constraints to specify endpoints of graph paths; the other constraints only specify properties of paths. Regular constraints specify paths using regular expressions, adapted to deal with multiple paths and infinite alphabets. 
Arithmetical constraints compare linear combinations of aggregated values, i.e., values of \emph{labels} accumulated along whole paths.

The language \pra{} can only aggregate and compare the values of labelling functions already defined in the graph. The properties we are interested in often require performing some arithmetical operations on the labellings, either simple (taking a linear combination) or complicated (taking minimum, maximum, or even computing some subquery).
Such operations are often nested inside regular expressions (as in \LARE{} \cite{Grabon+16}) making queries unnecessarily complicated. Instead, similarly as in \cite{Arenas+2014} we specify such operations in a modular way as \emph{ontologies}. This leads to the language \opra{}, which comprises Ontologies and \pra{}.
In our approach all knowledge on graph nodes is encoded by labellings, and our ontologies also are defined as the auxiliary labellings. 
For example, having a labelling $\textrm{child}(x, y)$ stating that $x$ is a child of $y$, we can define a labelling $\textrm{descendant}(x, y)$ stating that $x$ is a descendant of $y$. 
Such labellings can be computed on-the-fly during the query evaluation. 

\Paragraph{Related work} 
\emph{Regular Path Queries (RPQs)} \cite{sigmod/CruzMW87,kr/CalvaneseGLV00} are usually used as a basic construction for graph querying.
RPQs are of the form $\piszczalka{x}{\pi}{y} \wedge \pi \in L(e)$ where $e$ is a (standard) regular expression. 
Such queries return pairs of nodes $(v, v')$ connected by a path $\pi$ such that the labelling of $\pi$ forms a word from $L(e)$. 
\emph{Conjunctive Regular Path Queries (CRPQs)} are the closure of RPQs under conjunction and existential quantification \cite{pods/ConsensM90,siamcomp/MendelzonW95}.
Barcelo et al., \cite{Barcelo+12} introduced \emph{extended CRPQs (ECRPQs)} that can compare tuples of paths by \emph{regular relations} \cite{Elgot1965,tcs/FrougnyS93}. Examples of such relations are path equality, length comparisons, prefix (i.e., a path is a prefix of another path) and fixed edit distance. 
Regular relations on tuples of paths can be defined by the standard regular expressions over alphabet of tuples of edge symbols. 

Graph nodes often store \emph{data values} from an infinite alphabet. In such graphs, paths are interleaved sequences of data values and edge labels. 
This is closely related to \emph{data words} studied in XML context \cite{NevenSV04,DemriLN07,Segoufin06,BojanczykDMSS11}.
  Data complexity of query answering for most of the formalisms for data words is NP-hard~\cite{Libkin+16}.
This is not, however, the case for \emph{register automata} \cite{tcs/KaminskiF94}, which inspired Libkin and Vrgo\v{c} to define \emph{Regular Queries with Memory (RQMs)} ~\cite{Libkin+16}. RQMs are again of the form $\piszczalka{x}{\pi}{y} \wedge \pi \in L(e)$. However, $e$ is now \emph{Regular Expression with Memory (REM)}. REMs can store in a register the data value at the current position and test its equality with other values already stored in registers. 
Register Logic~\cite{BarceloFL15} is, essentially, the language of REMs closed under Boolean combinations, node, path and register-assignment quantification. It allows for comparing data values in different paths. The positive fragment of Register Logic, RL$^+$, has data complexity in \NL{}, even when REMs can be nested using a branching operator.
Walk Logic~\cite{HellingsKBZ13} extends FO with path quantification and equality tests of data values on paths. Query answering for WL is decidable but its data complexity is not elementary~\cite{BarceloFL15}. \LARE{} \cite{Grabon+16} is a query language that can existentially quantify nodes and paths, and check relationship between many paths. Path relationships are defined by 
 regular expressions with registers that allow for various arithmetic operations on registers.    

\begin{figure}
\centering
\begin{tikzpicture}[->,>=stealth',shorten >=1pt,auto,node distance=2cm,semithick]
  \tikzstyle{every state}=[fill=white,draw=black,text=black,minimum size=0.9cm]

  \node[state, ellipse] (A) at (0,0)                   {ECRPQ};
  \node[state, ellipse] (Z) at (-3.0,0)        {RPQ};
  \node[state, ellipse] (B) at (3,2) {ECRPQ+LC};
  \node[state, ellipse] (C) at (3, 0) {LARE};
  \node[state, ellipse] (D) at (6, 0) {\opra{}};

  \path (A) edge              node {$\subsetneq$} (B)
            edge              node {$\subsetneq$} (C)
        (B) edge              node {$\subsetneq$} (D)
	    (C) edge 		      node {$\subsetneq$} (D)
        (Z) edge              node {$\subsetneq$} (A)
        ;
\end{tikzpicture}
\caption{Comparison between different query languages}\label{fig:diagram}
\end{figure}
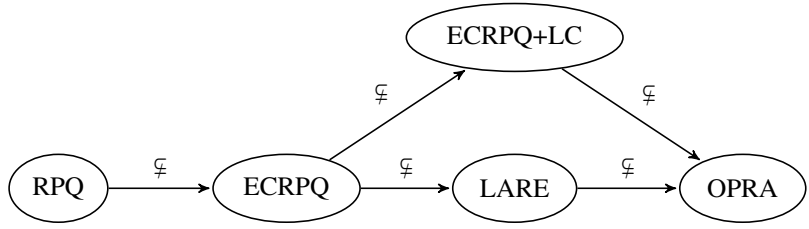

\Paragraph{Aggregation}
Ability to use aggregate functions such as sum, average or count is a~fundamental mechanism in database systems.
Klug \cite{Klug82} extended the relational algebra and calculus with aggregate functions and proved their equivalence. 
Early graph query languages $G^+$~\cite{CruzMW88} or GraphLog~\cite{pods/ConsensM90,ConsensM93} can aggregate data values.  
Consens and Mendelzon~\cite{ConsensM93} studied  \emph{path summarization}, i.e., summarizing information along paths in graphs.
They assumed natural numbers in their data model and allowed to aggregate summarization results.
In order to achieve good complexity (in the class NC) they allowed aggregate and summing operators that form a closed semiring. 
Other examples of aggregation can be found in \cite{Wood12}.

Summing vectors of numbers along graph paths have been already studied in the context of various formalisms based on automata or regular expressions
and lead to a number of proposals that have combined complexity in PSPACE and data complexity in \NL{}. 
Kopczy\'nski and To \cite{KopczynskiT10} have shown that \emph{Parikh images} (i.e., vectors of letter counts) for the usual finite automata can be expressed
using a union of linear sets that is polynomial in the size of the automaton and exponential in the alphabet size (the alphabet size, in our context, corresponds to the dimension of vectors).  
Barcelo et al. \cite{Barcelo+12} extended ECRPQs with linear constraints on the numbers of edge labels counts along paths.
They expressed the constraints using reversal-bounded counter machines, translated further to Presburger arithmetic formulas of a polynomial size and evaluate them using techniques from \cite{KopczynskiT10,Scarpellini84}. 

Figueira and Libkin \cite{FigueiraLibkin15} studied \emph{Parikh automata} introduced in 
\cite{KlaedtkeR03}. These are finite automata that additionally store a vector of \emph{counters} in $\mathbb{N}^k$
. Each transition specifies also a vector of natural numbers.     
While moving along graph paths according to a transition the automaton adds this transition's vector to the vector of counters.
The automaton accepts if the computed vector of counters is in a given semilinear set in $\mathbb{N}^k$. 
Also a variant of regular expressions capturing the power of these automata is shown.  
This model has been used to define a family of variants of CRPQs that can compare tuples of paths using \emph{synchronization languages} \cite{FigueiraL15}. 
This is a relaxation of regularity condition for relations on paths of ECRPQs and 
leads to more expressive formalisms with data complexity still in \NL{}.
These formalisms are incomparable to ours since they can express non-regular relations on paths like suffix but cannot express properties of data values, nodes' degrees or extrema.

Cypher \cite{TheNeo4jTeam17} is a practical query language implemented in the graph database Neo4j.
It uses \emph{property graphs} as its data model. These are graphs with labelled nodes and edges, but edges and nodes can also store attribute values for a set of \emph{properties}.
\texttt{MATCH} clause of Cypher queries allows for specifying graph patterns that depend on nodes' and edges' labels as well as on their properties values.
Cypher does not allow full regular expressions however graph patterns can contain transitive closure over a single label.
More on Cypher can be found in a survey \cite{Angles+16}.

RDF \cite{Cyganiak+2014} is a W3C standard that allows encoding of the content on the Web in a form of a set of \emph{triples} representing an edge-labelled graph.
Each triple consists of the subject $s$, the predicate $p$, and the object $o$ that 
are resource identifiers (URI’s), and represents an edge from $s$ to $o$ labelled by $p$. 
Interestingly, the middle element, $p$, may play the role of the first or the third element of another triple.  
Our formalism \opra{} allows to operate directly on RDF without any complex graph encoding, by using a ternary labelling representing RDF triples. 
This allows for convenient navigation by regular expressions in which also the middle element of a triple can serve as the source or the target of a single navigation step (cf. \cite{Libkin+13}). 
The standard query formalism for RDF is SPARQL \cite{PrudhommeauxS08,HarrisS+13}. It implements \emph{property paths} which are RPQs extended with inverses and limited form of negation (see survey \cite{Angles+16}).

 \section{Language \opra{}}
Various kinds of data graphs are possible and presented in the literature. The differences typically lie in the way the elements of graphs are labelled -- both  nodes and edges may be labelled by finite or infinite alphabets, which may have some inner structure. Here, we choose a general approach in which a \emph{labelled graph}, or simply a graph, is a tuple consisting of a finite number of \emph{nodes} $V$ and a number of labelling functions $\lambda: V^l \to \Z\cup\{-\infty, \infty\}$ assigning integers to vectors of nodes of some fixed size.
While edges are not explicitly mentioned, if needed, one can consider an \emph{edge} labelling $\lambda_E$ such that $\lambda_E(v, v')$ is $1$ if there is an edge from $v$ to $v'$
and it is $0$ otherwise. More sophisticated edges, e.g., with integer labels, may be handled by means of standard embedding, defined in Section \ref{sec:ep}.
For convenience, we assume that the set of nodes always contains a distinguished node $\padding$ -- we  use it as a ``sink node'', to avoid problems with paths of different lengths.

A \emph{path} is a sequence of nodes. For a path $p=v_1 \dots v_k$, by $p[i]$ we denote its $i$-th element, $v_i$, if $i\leq k$, and $\padding$ otherwise.

\subsection{Basic constructs}
We first define the language \pra{}, which is the core of the language \opra{}.
The queries of \pra{} are of the form \smallskip\\
\q{MATCH NODES $\vec{x}$, PATHS $\vec{\pi}$\\
SUCH THAT \text{P}ath\_constraints\\
WHERE \text{R}egular\_constraints\\
HAVING \text{A}rithmetical\_constraints} \smallskip\\
 where 	$\vec{x}$ are free node variables, $\vec{\pi}$ are free path variables,
\q{Path\_constraints} is a~conjunction of path constraints, \q{Regular\_constraints} is a conjunction of \emph{regular constraints} and, finally, \q{Arithmetical\_constraints} is a conjunction of \emph{arithmetical constraints}, as defined below. The constraints may contain variables not listed in the \q{MATCH} clause (which are then existentially quantified).

\Paragraph{Path constraints} Path constraints are expressions of the form $\piszczalka{x_s}{\pi}{x_t}$, where $x_s, x_t$ are node variables and $\pi$ is a path variable, satisfied if $\pi$ is a sequence of nodes starting from $x_s$ and ending in $x_t$.

\Paragraph{Regular constraints}
The main building blocks of regular constraints are \emph{node constraints}. 
Syntactically, a $k$-node constraint is an expression 
containing free node variables $\z_1, \z_1', \dots, \z_k, \z_k'$ and of the form $X \op X'$, where 
$\op \in \{\leq, <, =\}$ and each of $X, X'$ is an integer constant or a labelling function $\lambda_i$ applied to some of the free variables. 

A $k$-node constraint for a regular constraint over $k$ paths 
may be seen as a function that takes a vector containing two nodes of each path: a \emph{current node} (represented by $\z_i$) and a \emph{next node} (represented by $\z_i'$), and returns a Boolean value. 
The semantics is given by applying the appropriate labelling functions to the nodes given as an input and comparing the value  according to the $\op$ symbol. 

A regular constraint $R(\pi_1, \ldots, \pi_k)$ is syntactically a regular expression over an infinite alphabet consisting of all the possible $k$-node constraints. 
Assume $p_1, \dots, p_k$ are paths and let $w_1 \dots w_s$ be the word such that $s = \max(|p_1|, \dots, |p_k|)$ and each $w_i \in V^{2k} $ is defined as $w_i = (p_1[i], p_1[i+1], \ldots, p_k[i], p_k[i+1])$, i.e., it is a vector consisting of $i$-th and $i+1$-th node of each path (or can be substituted by $\padding$ if not present). We say that $p_1, \dots, p_k$
 \emph{satisfy} $R$, denoted as $R(p_1, \dots, p_k$), if the $w_1 \dots w_s$ belongs to the language $L^G(R)$, defined inductively in the usual manner:
\begin{compactitem}
\item $L^G(R)$, where $R$ is a node constraint, is defined as a set of vectors of length $2k$ for which the constraint $R$ returns true.
\item $L^G(R \cdot R')=\set{a \cdot b \mid a \in L^G(R) \land b \in L^G(R')}$.
\item $L^G(R+R')$ is the union of $L^G(R)$ and $L^G(R')$.
\item $L^G(R^*)$ is the Kleene-star closure of $L^G(R)$.
\end{compactitem}

\sloppy
Note that acccording to the definitions above, the variables $\z_1, \z_1', \ldots, \z_k, \z_k'$ in a regular constraint $R(\pi_1, \ldots, \pi_k)$ always refer to the nodes of the paths $\pi_1, \ldots, \pi_k$, e.g., $\z_4$ refers to the current node of the path $\pi_4$ and $\z_2'$ refers to the next node of the path $\pi_2$. In order not to mix the variables with ordinary ones we disallow to use them in any other context.  

\fussy
\Paragraph{Arithmetical constraints} An arithmetical constraint is an inequality 
\( c_{1} \Lambda_1 + \ldots + c_j \Lambda_j\leq c_0 \),
where $c_0, \ldots, c_j$ are integer constants and each $\Lambda_\ell$ is an expression of the form $\lambda_i[\pi_{i_1}, \ldots, \pi_{i_k}]$. 
The value of $\lambda_i[\pi_{i_1}, \ldots, \pi_{i_k}]$ 
over paths $p_1, \dots, p_n$ is defined as the sum $\sum_{i=1}^{s} \lambda_i(p_{i_1}[i], \ldots, p_{i_k}[i])$, where $s=\max\{|p_1|, \dots, |p_k|\}$, i.e., the sum of the labelling for vectors of nodes on corresponding positions of all paths. 
Paths $\vec{p}$ satisfy the arithmetical expression $c_{1} \Lambda_1 + \ldots + c_j \Lambda_j \leq c_0$ if
the value of the left hand side, with $\vec{\pi}$ instantiated to $\vec{p}$, is less than or equal to $c_0$.

\Paragraph{Query semantics} 
Let $Q(\vec{x}, \vec{p})$ be a \pra{} query, and $\vec{x}'$ and $\vec{\pi}'$ be node and path variables in $Q$ that are not listed as free.
We say that nodes $\vec{v}$ and paths $\vec{p}$ (of some graph $G$) satisfy $Q$, denoted as $Q(\vec{v}, \vec{p})$,  if and only if there exist nodes $\vec{v}'$ 
and paths $\vec{p}'$ such that the instantiation
$\vec{x} = \vec{v}$, $\vec{x}'  = \vec{v}'$, 
$\vec{\pi} = \vec{p}$ and $\vec{\pi}'  = \vec{p}'$ satisfies all constraints in $Q$.

\subsection{Auxiliary labelling}
We introduce a way of defining auxiliary labellings of graphs, which are defined based on existing graph labellings and its structure. The ability to define auxiliary labellings significantly extends the expressive power of the language.
The essential property of auxiliary labellings is that their values do not need to be stored in the database, which would require polynomial memory, but can be computed \emph{on demand}. An auxiliary labelling may be seen as an \emph{ontology} or a \emph{view}.

We assume a set $\F$ of \emph{fundamental functions} $f:(\Z\cup\{-\infty, \infty\})^* \to \Z\cup\{-\infty, \infty\}$ consisting of \emph{aggregate functions} maximum $\fmax$, minimum $\fmin$,
counting $\fcount$, summation $\fsum$, and \emph{binary} functions $+$, $-$, $\cdot$ and $\leq$ (assuming $0$ for false and $1$ for true, and that these functions return $0$ if the number of inputs is not two).
$\F$ can be extended, if needed, by any functions computable by a non-deterministic Turing machine whose size of all tapes while computing $f(\vec{x})$  is 
logarithmic in length of $\vec{x}$ and values in $\vec{x}$, assuming binary representation, provided that additional aggregate functions in $\F$ are invariant under permutation of arguments.

\Paragraph{Terms} 
In order to specify values for auxiliary labellings we use \emph{terms}.
A term $t(\vec{x})$ is defined by the following BNF:
\begin{align*}
t(\vec{x}) ::=  c & 
\mid \lambda(\vec{y}) 
\mid [Q(\vec{y})] 
\mid \min_{\lambda, \pi} Q(\vec{y}, \pi) 
\mid \max_{\lambda, \pi} Q(\vec{y}, \pi) 
\\ &
\mid y = y
\mid  f(t(\vec{y}), \dots, t(\vec{y}) ) 
\mid f'(\{ t(x) \colon t(x,\vec{y})\})
\end{align*}
where $\vec{x}$ is a vector of node variables, $x$ is a fresh node variable, $c$ is a constant,
$\lambda$ is a labelling, $Q$ is a \pra{} query in $G$, $f \in \F$, $f'\in \F$ is aggregate, $\vec{y}$ ranges over vectors of variables among $\vec{x}$ and $y$ ranges over variables among $\vec{x}$.

Let $G$ be a graph. A \emph{variable instantiation}  $\eta^G: \vec{x} \to V$ in $G$  is a function that maps variables in $\vec{x}$ to nodes of $G$. 
Such a function extends canonically to subvectors of $\vec{x}$. 
Below we inductively extend variable instantiations to terms.
If $G$ is clear from the context, we write $t(\vec{v})$ as a shorthand of $\eta^G(t(\vec{x}))$, where $\eta^G(\vec{x}[i])=\vec{v}[i]$ for all $i$.
\begin{enumerate}[ 1.]
\item $\Tvalue{c} = c$, where $c$ is a constant,
\item $\Tvalue{\lambda(\vec{y})} = \lambda(\eta^G(\vec{y}))$, where $\lambda$ is a labelling of $G$
\item $\Tvalue{[Q(\vec{y})]}$ is $1$ if $Q(\eta^G(\vec{y})
)$ holds in $G$ and $0$ otherwise,
\item 
$\Tvalue{\min_{\lambda, \pi} Q(\vec{y}, \pi)}$ is the 
minimum of values of $\lambda[p]$, defined as in the arithmetical constraints, over all paths $p$ such that 
$Q(\eta^G(\vec{y}),p)$ holds in $G$, 
\item 
$\Tvalue{\max_{\lambda, \pi} Q(\vec{y}, \pi)} = \max(\set{\lambda[p] \mid Q(\eta^G(\vec{y}),p)})$, 
\item $\Tvalue{y=y'}$ is $1$ if $\Tvalue{y}=\Tvalue{y'}$ and $0$ otherwise,
\item  $\Tvalue{f(t_1(\vec{y}_1), \ldots, t_k(\vec{y}_k)} {=} f(\Tvalue{t_1(\vec{y}_1},\ldots, \Tvalue{t_k(\vec{y}_k}))$,  
\item  
\(\Tvalue{f(\{ t(x) {:} t'(x,\vec{y}) \})}=f(t(v_1), \ldots, t(v_n))\),
where $v_1, \ldots, v_n$ are all nodes $v$ of $G$ s.t. $t'(v, \eta^G(\vec{y}))=1$. 
\end{enumerate}

\Paragraph{Auxiliary labellings}
Consider a term $t(\vec{x})$ and a graph $G$, which does not have a labelling $\lambda$. 
We define the graph $G[\definedBy{\lambda}{t}]$ as the graph $G$ extended with the labelling $\lambda$ 
such that 
$\lambda(\vec{v}) = t(\vec{v})$ for all $\vec{v}\in V^k$.
We call $\lambda$ an \emph{auxiliary labelling} of $G$.
We write  $G[\definedBy{\lambda_1}{t_1}, \ldots, \definedBy{\lambda_n}{t_n}]$ to denote 
the results of successively adding labellings $\lambda_1, \dots, \lambda_{n}$ to the graph $G$, i.e., $G[\definedBy{\lambda_1}{t_1}][\definedBy{\lambda_2}{t_2}]\ldots [\definedBy{\lambda_n}{t_n}]$.

\Paragraph{Language \opra{}} An \opra{} query is an expression of the form
\q{LET $O$ IN $Q'$}, where $Q'$ is a \pra{} query, $O$ is of the form $\definedBy{\lambda_1}{t_1}, \ldots, \definedBy{\lambda_n}{t_n}$ and $t_1, \ldots, t_n$ are terms. 
The query $Q$ holds over graph $G$, nodes $\vec{v}$ and paths $\vec{p}$, denoted as $Q(\vec{v}, \vec{p})$, if and only if $Q'(\vec{v}, \vec{p})$ holds over 
$G[O]$.
Note that $Q'(\vec{x})$ can refer to auxiliary labellings $\lambda_{1}, \ldots, \lambda_{n}$.
The size of $Q$ is the sum of binary representations of terms $t_1, \ldots, t_n$
and the size of query $Q'$. 
 
 \section{Examples}
\label{sec:examples}
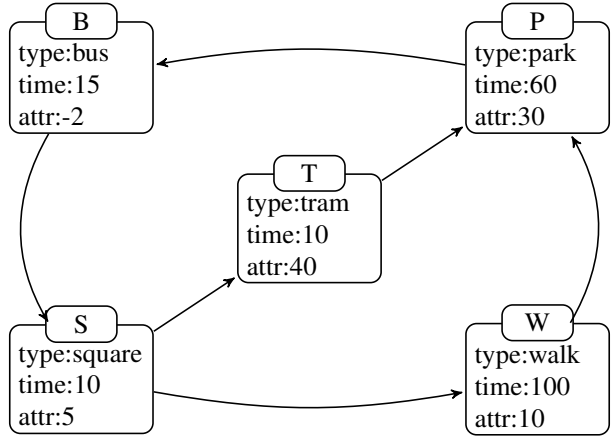
\begin{figure}
\centering
\begin{tikzpicture}[->,>=stealth',shorten >=1pt,auto,node distance=2.9cm,semithick]
  \tikzstyle{every state}=[fill=white,draw=black,text=black,minimum size=0.8cm]
  \tikzset{
   datanode/.style = {
    draw, 
    rectangle,
    rounded corners,
    align=left,
    text width=1.65cm,
    inner ysep = +0.8em},
    labelnode/.style = {
    draw, 
    rectangle,
    rounded corners,
    align=center,
    fill=white}
  }

\foreach \lab\lX/\lY/\name/\kind/\cost/\time/\attr in 
         {A/0/0/S/square/0/10/5,
          AB1/3/2/T/tram/5/10/40,
          B/6/4/P/park/100/60/30,
          AB2/6/0/W/walk/0/100/10,
          BA/0/4/B/bus/10/15/-2} 
{
  
  \node[datanode] (\lab) at (\lX,\lY)  
   {
      type:\kind\newline time:\time\newline attr:\attr \vspace{-5pt}
   };
  \node[labelnode,fill=white,text width=0.7cm] (small\lab) at (\lab.north)  
  {
   \name
  };
  
}

\path (A) edge (AB1)
      (AB1) edge (B) 
      (A) edge[out=350,in=190] (AB2)
      (AB2) edge[bend right] (B)
      (B) edge[out=170,in=10] (BA)
      (BA) edge[bend right] (A)
            
;

\end{tikzpicture}
\caption{An example of a map-representing graph}\label{fig:map}
\end{figure}

We focus on the following scenario: a graph database that corresponds to a map of some area. 
Each graph's node represents either a \emph{place} or a \emph{link} from one place to another. 
The graph has four unary labellings and one binary labelling. 
The labelling  $\labelling{type}$ represents the type of a place for places (e.g., square, park, pharmacy) or the mode of transport for links (e.g., walk, tram, train); we assume each type is represented by a constant, e.g., $c_{\textrm{square}}$, $c_{\textrm{park}}$.
 The labelling $\labelling{attr}$ represents attractiveness  (which may be negative, e.g., in unsafe areas), 
and $\labelling{time}$ represents time. 
The binary labelling $\labelling{E}$ represents edges: for nodes $v_1, v_2$, the value $\labelling{E}(v_1,v_2)$ is $1$ if there is an edge from $v_1$ to $v_2$
and $0$ otherwise. 
For example, the graph on Fig. \ref{fig:map} represents a map with two places: $S$ is a square and $P$ is a park. There are three nodes representing links: node $W$ represents moving from $S$ to $P$ by walking, $T$ moving from $S$ to $P$ by tram and $B$ moving from $P$ to $S$ by bus.

\subsection{Language \pra} \label{subsec:pra}
We begin with the query $Q_\mathrm{route}(s, t, \pi)$ stating that there is a path $\pi$ from a node $s$ to a node $t$ such that 
each pair of consecutive nodes on this path is connected by an edge given by the edge labelling $\labelling{E}$.
Recall that our path constraints of the form $\piszczalka{s}{\pi}{t}$ require only that $\pi$ is a sequence of nodes that starts at $s$  and ends at $t$. It does not depend on any labelling, in particular $\labelling{E}$. 
We introduce a regular constraint \q{$\Path(\pi)$}
defined as 
\q{$\langle \labelling{E}(\z_1, \z'_1) = 1\rangle^\ast\langle\top\rangle(\pi)$} that states that any two consecutive nodes on $\pi$ satisfy $\labelling{E}$. As the last node has no successor (i.e., $\z_1'=\padding$ for the last node), the constraint ends with $\langle\top\rangle$ that is always satisfied. Then, we can express $Q_\mathrm{route}(s, t, \pi)$ as

\noindent \q{MATCH NODES $(s, t)$ 
SUCH THAT $\piszczalka{s}{\pi}{t}$ WHERE $\Path(\pi)$}

\begin{Property}{Sums}
The language \pra{} can express properties of paths' sums. For example, the query below holds iff there is a route from $s$ to $t$ that takes at most 6 hours and its attractiveness is over 100.

\noindent \q{MATCH NODES $(s, t)$ 
SUCH THAT $\piszczalka{s}{\pi}{t}$ WHERE $\Path(\pi)$\\
HAVING $\labelling{time}[\pi] \leq 360 \wedge \labelling{attr}[\pi] > 100$}

Furthermore, we can compute averages, to some extent. For example, the following arithmetical constraint says that for some path $\pi$ the average attractiveness of $\pi$ is at least 4 attractiveness points per minute: \q{$\labelling{attr}[\pi] \geq 4 \labelling{time}[\pi]$}.
\end{Property}

\begin{Property}{Multiple paths}
We define a query that asks whether there is a route from $s$ to $t$, such that from every place we can take a tram (e.g., if it starts to rain).
We express that by stipulating a route $\pi$ from $s$ to $t$ and a sequence $\rho$ of tram links, such that every node of $\pi$
representing a place is connected with the corresponding tram link in $\rho$. In a way, $\rho$ works as an existential quantifier for nodes of $\pi$.
 
\noindent \q{MATCH NODES $(s, t)$ 
SUCH THAT $\piszczalka{s}{\pi}{t}$  \\ WHERE  $\Path(\pi) \wedge 
\langle \labelling{type}(\z_1) = c_{\mathrm{tram}} \rangle^\ast (\rho) \wedge
\mathrm{Link}(\pi, \rho)
 $} \\
where \q{$\mathrm{Link}=
(\langle \labelling{type}(\z_1) = c_{\mathrm{bus}} \rangle + \langle \labelling{type}(\z_1) = c_{\mathrm{walk}} \rangle + \langle \labelling{type}(\z_1) = c_{\mathrm{tram}} \rangle +\langle \labelling{E}(\z_1, \z_2) = 1 \rangle )^\ast{}$} states that every node of the first path
either is not a place, i.e, it represents any of possible links (by a bus, a walk or a tram), or is connected with the corresponding node of the second path.  
Note also that in the regular constraint $\langle \labelling{type}(\z_1) = c_{\mathrm{tram}} \rangle^\ast (\rho)$
the variable $\z_1$ represents the current node of the path $\rho$, whereas, in  
$\mathrm{Link}(\pi, \rho)$ the variable $\z_1$ represents the current node of $\pi$, and $\z_2$ represents the current node of $\rho$.
\end{Property}

\subsection{Language  \opra}
\label{s:opra}
We show how to employ auxiliary labellings in our queries. 
For readability, we introduce some syntactic sugar -- constructions which do not change the expressive power of \opra, but allow queries to be expressed more clearly.
We use the function symbols $=, \neq$ and Boolean connectives, which can be derived from $\leq$ and arithmetical operations. Also, 
we use terms $t(x,y)$ in arithmetical constraints, which can be expressed by first defining the labelling $\definedBy{\labelling{t}(x,y)}{t(x,y)}$,
defining additional paths $\rho_1 = x,\rho_2 = y$ of length $1$, and using $\labelling{t}[\rho_1, \rho_2]$.

\begin{Property}{Processed labellings}
Online route planners often allow to look for routes which do not require much walking.
The following query asks whether there exists a route from $s$ to $t$ such that the total walking time is at most 10 minutes. 
To express it, we define a labelling $\labelling{t\_walk}(x)$, which is the time of $x$ for $x$ that are walking links, and $0$ otherwise.

\noindent \q{LET \(\labelling{t\_walk}(x) := (\labelling{type}(x) = c_\textrm{walk}) \cdot \labelling{time}(x)\) IN\\
MATCH NODES $(s, t)$ SUCH THAT $\piszczalka{s}{\pi}{t}$ \\
WHERE $\Path(\pi)$ 
HAVING \(\labelling{t\_walk}[\pi] \leq 10\)
}
\end{Property}

\begin{Property}{Nested queries}
It is often advisable to avoid crowded places, which are usually the most attractive places. 
We write a query that holds for routes that are always at least 10 minutes away from any node with attractiveness greater than $100$.
We define a labelling  $\labelling{crowded}(x)$ as

\noindent \q{[MATCH NODES $(x)$ SUCH THAT  $\piszczalka{x}{\pi}{y}$  WHERE \\$\Path(\pi) \wedge \langle \top \rangle^\ast \langle \labelling{attr}(\z_1) > 100 \rangle(\pi)$   HAVING  $\labelling{time}[\pi] \leq 10$]}

\noindent Notice that $\pi$ and $y$ are existentially quantified. We check whether the value of \q{$\labelling{crowded}$} is $0$ for each node of the path $\pi$.

\noindent \q{MATCH PATHS $(\pi)$ 
 WHERE $ \Path(\pi) \wedge \langle \labelling{crowded}(\z_1) = 0\rangle^\ast(\pi)$}
\end{Property} 
 
\begin{Property}{Nodes' neighbourhood}
``Just follow the tourists'' is an advice given quite often. With \opra{}, we can verify whether it is a good advice in a given scenario. 
A route is called \emph{greedy} if at every position, the following node on the path is the most attractive successor.
We define a labelling \q{$\labelling{MAS}(x,y)$} that is $1$ if $y$ is the most attractive successor of $x$, and $0$ otherwise: 
\q{\( \labelling{E}(x, y) \land (\fcount(\{  \labelling{attr}(z) \colon \labelling{E}(x, z) \land \labelling{attr}(z) {\geq} \labelling{attr}(y) \})  {=} 1). \) }
 We express that there is a greedy route from $s$ and $t$.\\ 
\q{MATCH NODES \((s, t) \) 
SUCH THAT \(\piszczalka{s}{\pi}{t}\)\\
WHERE \( \langle \labelling{MAS}(\z_1, \z_1') = 1\rangle^\ast\langle\top\rangle(\pi) \)
}
\end{Property}

\begin{Property}{Properties of paths' lengths}
In route planning, we often have to balance time, money, attractions, etc. The following query asks
whether is it possible to get from $s$ to $t$ in a shortest time possible, in the same time maximising the attractiveness of the route.
Recall that $Q_\mathrm{route}$ is a \pra{} query defined in Subsection \ref{subsec:pra}.

\noindent \q{MATCH NODES $(s, t)$ SUCH THAT  $\piszczalka{s}{\pi}{t}$  
WHERE $\Path(\pi)$ \\
HAVING
$(\labelling{attr}[\pi] = \max_{\labelling{attr}, \rho} Q_\mathrm{route}(s, t, \rho)) \wedge\\
(\labelling{time}[\pi] = \min_{\labelling{time}, \rho} Q_\mathrm{route}(s, t, \rho))$}
\end{Property} 

\begin{Property}{Registers}
\sloppy
Registers are an important concept often in graph query languages. 
For instance, to express that two paths have a non-empty intersection, we load a (non-deterministically picked) node from the first path to a register and check whether it occurs in the second path. 
The following query asks whether there exists a route from a club $s$ to a club $t$ on which the attractiveness of visited clubs never decreases.
In the register-based approach, we achieve this by storing the most recently visited club in a separate register. 
Here, we express this register using an additional path $\rho$, storing the values of the register, and a labelling
$\labelling{r}(x', y, y')$ which states that $y' = x'$ if $x'$ is a club, and $y'=y$ otherwise, defined as
 \q{${(\labelling{type}(x') = c_\textrm{club} \Rightarrow y' = x') \land (\labelling{type}(x') \neq c_\textrm{club} \Rightarrow y=y')}$}.

\noindent\q{MATCH NODES $(s, t)$ 
SUCH THAT \(\piszczalka{s}{\pi}{t} \land \piszczalka{s}{\rho}{t}\) \\
WHERE 
\(\Path(\pi) \land \mathrm{ends}(\pi) \land \mathrm{regs}(\pi, \rho) \land \mathrm{inc}(\rho)\)}\\
where 
\q{$\mathrm{regs}=
\langle \labelling{r}(\z_1',\z_2, \z_2') = 1 \rangle^\ast{}\langle \top \rangle$} ensures that at each position the second path contains the most recently visited club along the first path, 
\q{$\mathrm{ends} = \langle \labelling{type}(\z_1)=c_\textrm{club} \rangle \langle \top \rangle^\ast{} \langle \labelling{type}(\z_1) = c_\textrm{club} \rangle$} states that the both ends of a path are clubs, 
and \q{\( \mathrm{inc}=
 \langle\labelling{attr}(\z_1) \leq \labelling{attr}(\z_1') \rangle^\ast{}\langle \top \rangle
\)} checks that the attractiveness never decreases.
\fussy
\end{Property}

 \section{Expressive power}
\label{sec:ep}
\sloppy
We compare the expressive power of \opra{} and other query languages for graph databases from the literature. We prove the results depicted in Figure \ref{fig:diagram}: that \opra{} subsumes ECRPQ, ECRPQ with linear constraints~\cite{Barcelo+12} and LARE~\cite{Grabon+16} query languages. 

\fussy

These query languages assume a different notion of graphs from the one considered in this paper. 
We call graphs as defined in these papers \emph{data graphs}. 
A \emph{data graph} is a tuple 
\(
G = \tuple{V, E, \lambda}
\)
where
$V$ is a finite set of nodes,
$E \subseteq V \times \Sigma \times V$ is a set of edges labelled by a finite alphabet $\Sigma$,
and $\lambda: V \to \Z^K$ is a labelling of nodes by vectors of $K$ integers. 
A path in $G$ is a sequence of interleaved nodes and edge labels $v_0 e_1 v_1 \ldots v_k$ such that for every $i<k$ we have $E(v_{i}, e_{i+1}, v_{i+1})$.

The difference between graphs and data graphs is mostly syntactical, yet it prevents us from comparing directly the languages of interest. 
To overcome this problem, we define the \emph{standard embedding}, which is a natural transformation of data graphs to graphs. 
For a data graph $G = \tuple{V, E, \lambda}$ with edges labelled by $\Sigma$ and nodes labelled by $\Z^K$, we define the graph $G^E = (V^E, \lambda_1^E, \ldots, \lambda_K^E, \lambda_{K+1}^E)$, called the standard embedding of $G$, such that
(1)~$V^E = V \cup \Sigma$, 
(2)~for every $i \in \{1, \ldots, K\}$ and every $v \in V$ we have $\lambda_i^E(v)$ equal to the $i$-th component of $\lambda(v)$,
(3)~for every $i \in \{1, \ldots, K\}$ and every $v \in \Sigma$ we have   $\lambda_i^E(v) = 0$, and
(4)~for all $v_1, v_2, v_3 \in V^E$ we have $\lambda^E_{K+1}(v_1, v_2, v_3) = 1$ if $(v_1,v_2, v_3) \in E$ and  $\lambda^E_{K+1}(v_1, v_2, v_3) = 0$ otherwise.
Observe that every node $v$ (resp., every path $p$) in $G$ corresponds to the unique node $v^E$ (resp., path $p^E$) in $G^E$.

A query $Q_1$ on data graphs is equivalent w.r.t. the standard embedding, \emph{se-equivalent} for short, to a query $Q_2$ on graphs if 
for all data graphs $G$, nodes $\vec{v}$ in $G$ and paths $\vec{p}$ 
query $Q_1(\vec{v}, \vec{p})$ holds in $G$ if and only if $Q_2(\vec{v}^E, \vec{p}^E)$ holds in $G^E$, where
$G^E, \vec{v}^E, \vec{p}^E$ result from the standard embedding of respectively $G, \vec{v}, \vec{p}$.
We say that \opra{} \emph{subsumes} a query language $\lang$ if every query in $\lang$ can be transformed in polynomial time
to an se-equivalent \opra{} query.

\Paragraph{LARE}
Queries in LARE are built from arithmetical regular expressions, which extend regular expressions with registers storing nodes and
arithmetical operations on labels of the nodes stored in registers (which are natural numbers). 
We briefly discuss how to express the three main building blocks of LARE expressions:
\emph{edge constraints} specifying labels of edges, \emph{register constraints} specifying values of registers,
and \emph{register assignments} specifying  how registers change. 

\opra{} queries over the standard embedding can specify labels of edges in the original graph and hence can express edge constraints. 
Next, we arithmetize all logical operations assuming \emph{true}:= $1$ and \emph{false}:= $0$.
With that, we can show by structural induction that \opra{} labellings can express register constraints (e.g., construction $C \vee C'$ can be expressed by 
term $\max(t_C, t_{C'})$).
Finally, we can express registers with additional paths as discussed in Section~\ref{s:opra}.
 
\opra{} language is stronger than LARE. The set of (vectors of) paths satisfying a given LARE query is regular. Therefore, 
for a fixed LARE query $Q(\pi_1, \pi_2)$, we can decide whether $\forall \pi_1 \exists \pi_2 Q(\pi_1, \pi_2)$ holds in a given graph $G$ in polynomial space in $G$.
The set of paths satisfying an \opra{} query $Q$ is also related to automata, but due to linear constraints we can express properties of weighted automata.
We can define an \opra{} query $Q^U(\pi_1, \pi_2)$ which interprets the input graph $G$ as a weighted automaton, path $\pi_1$ as an input word  
and $\pi_2$ as a run $r$ on $\pi_1$; query $Q^U$ holds only if the value of $r$ is at most $0$.
Then, $\forall \pi_1 \exists \pi_2 Q(\pi_1, \pi_2)$ holds only if the value of every word (w.r.t. the weighted automaton corresponding to $G$) is at most $0$. 
Such a problem for weighted automata is called 
(quantitative) universality problem and it is undecidable~\cite{Almagor+2011}. 
Therefore, checking whether a given graph $G$ satisfies $\forall \pi_1 \exists \pi_2 Q(\pi_1, \pi_2)$ is undecidable. 
Thus, no LARE query is se-equivalent to $Q^U$. 
 
\begin{theorem}
(1)~\opra{} subsumes \LARE{}.
(2)~There is an \opra{} query $Q$ with no LARE query $Q'$ se-equivalent to $Q$.
\end{theorem}
\begin{proof}(of (1))
\newcommand{\rx}{\alpha}
\newcommand{\ARE}{\textsc{are}}
\newcommand{\R}{{\mathbb{R}}}
\newcommand{\mygets}{\hspace{-1pt}\gets\hspace{-1pt}}
\newcommand{\N}{{\mathbb N}}
The whole proof is by induction on the nesting depth of a given LARE query $Q(x_1, \ldots, x_k)$. We assume that $Q(x_1, \ldots, x_k)$ has nesting depth $s$, and for all its 
nested queries there is an se-equivalent \opra{} query $Q'$.

The proof of Theorem 3 in~\cite{Grabon+16} starts with an observation that every LARE query $Q(x_1, \ldots, x_k)$ can be transformed to 
a query of the form $\piszczalka{y_1}{y_2}{y_3} \wedge \ldots \wedge \piszczalka{y_{m}}{y_{m+1}}{y_{m+2}} \wedge R_Q(x_1, \ldots, x_k)$, where 
$y_1, \ldots, y_{m+2} \in \{x_1, \ldots, x_k\}$, and the expression  $\piszczalka{y_1}{y_2}{y_3}$ denotes that $y_2$ is a path from $y_1$ to $y_3$, and
$R_Q(x_1, \ldots, x_k)$ is defined by some  $k$-ary arithmetical regular expression (\ARE{}) $\rx{}$. 
Observe that $\piszczalka{y_1}{y_2}{y_3}$ is expressed by the \opra{} query $Q_\mathrm{route}(s, t, \pi)$ defined in Section~\ref{subsec:pra}. 
Since \opra{} queries are closed under conjunction, it suffices to show that for every \ARE{} $\rx{}$ there exists a se-equivalent \opra{} query $Q^{\rx{}}$. 
To show that, we briefly recall the definition of \ARE{}.

\ARE{} are regular expressions with arithmetical functions and memory. The memory is formalized as an infinite set of registers $\R = \set{r_i \mid i \in \N}$, storing nodes.
Formally, $n$-ary \emph{arithmetical regular expressions} $\rx{}$ are defined in the following way:
\begin{align*}
  \rx{} ::= &\epsilon \mid \naww{C} \mid [r \mygets j] \mid \vec{e} \mid \rx{} + \rx{} \mid \rx{} . \rx{} \mid \rx{}^+ 
\end{align*}
where
$r$ ranges over $\R$,  and  $C$ ranges over \emph{register constraints}.

The expression  $[r \mygets j]$, read ``load $p_j$ to the register $r$'', states that the value of $r$ at the next step 
(past some letter $\vec{e}$) is equal to the current node stored on path $p_j$.
All the nodes of the paths have to be stored in registers  using the $[r \mygets j]$ syntax prior to their access. 
The expression $\naww{C}$ holds if the nodes stored in the registers satisfy the register constraint $C$. 
Finally, the expression $ \vec{e}$ denotes that the labels of the current edges on input paths $p_1, \ldots, p_n$
are respectively, $e_1, \ldots, e_n$, and $\vec{e} = (e_1,\ldots, e_n)$. This construct can be straightforwardly converted to an expression on 
labels in the standard embedding. 
In the following, we first discuss how can we implement registers, and then we define register constraints and discuss 
how can we express them through auxiliary labellings.

To simulate registers, for every register $r$ occurring in the \ARE{} $\rx{}$, we introduce a fresh path $p_r$. Then, we express register updates 
as in the example {Registers} in Section~\ref{s:opra}.
More precisely, for simplicity we discuss a transformation involving an exponential-blow up of the query.
Any \ARE{} can be considered as a regular expression over the alphabet of register constraints $\naww{C}$, register updates $[r \mygets j]$ and edge labels $\vec{e}$. 
Observe that $\rx{}$ can be transformed in polynomial time into another \ARE{} $\rx{}'$, in which all letters are combined into blocks of the form 
$\naww{C} [r_1 \mygets j_1] \ldots [r_m \mygets j_m]  \vec{e}$, where $r_1, \ldots, r_m$ are all registers. Basically, we need to transform $\rx{}$ into an automaton, and analyze all paths
between letters $\vec{e}$. Since checks of register constraints $\naww{C}$ and register updates $[r \mygets j]$ are commutative among one another, there are exponentially many different paths
between any two transitions involving letter $\vec{e}$. Moreover, a sequence of register constraints can be substituted by a single register constraint. Also,
if there is no register update for some register, we introduce an additional construct $[r \mygets 0]$, denoting that $r$ does not change its value. 
Next, we can build  \ARE{} $\rx{}'$, which is equivalent to $\rx{}$ and all its letters are of the form 
$\naww{C} [r_1 \mygets j_1] \ldots [r_m \mygets j_m]  \vec{e}$. For \ARE{} $\rx{}'$, for every sequence of register updates
$[r_1 \mygets j_1] \ldots [r_m \mygets j_m]$, we define an auxiliary labelling that 
enforces that the nodes at the paths in the next step are consistent with the sequence (see the example {Registers} in Section~\ref{s:opra}).

Now, we discuss how to express register constraints.
\emph{Register constraints} $C$ and \emph{arithmetical expressions} $P$ are defined in the following way:
\begin{align*}
  C ::= &C \lor C
  \mid \neg C 
  \mid \exists r . C
  \mid f(P, \dots, P) \leq f(P, \dots, P)\\
  &\mid r = r
  \mid E(r, e, r)
  \mid \query{Q}(r, \dots, r)\\
  P ::= &\lambda(r) \mid f[r : C]\\
  \end{align*}
where
$r$ ranges over the set of registers $\R$ and 
$f$ ranges over the set of fundamental functions $\F$ (as defined in this paper).

We briefly explain the above constructs. Register constraints are computed presuming some values of registers; they return Boolean values, and hence we can take disjunction $(C_1 \lor C_2)$ and negation ($\neg C'$)
of register constraints. We can also quantify existentially the values of registers ($\exists r . C$).
A register constraint can check inequality between the values of arithmetical expressions ($f_1(P_1, \dots, P_k) \leq f_2(P_1', \dots, P_l')$). 
In register constraints we can check whether two registers store the same node ($r = r'$) or whether
two registers are connected with an edge labelled by $e$ ($E(r,e,r')$). Finally, we can check whether a nested query holds by passing the values of some registers as the input to the query
 ($\query{Q}(r_1, \dots, r_l)$).

Arithmetical expressions return natural numbers presuming some values of registers. An arithmetical expression can either take the value of the labelling of the node stored in some register ($\lambda(r)$), or
it can compute a function $f$ applied to all nodes of the graph satisfying a register constraint $C$ ( $f[r : C]$).

We show that for every register constraint $C$ and arithmetical constraint $P$, there exist terms $t_C, t_P$ defining labellings equivalent to $C$ and $P$
, i.e.,
the register constraint $C$ holds (resp,. an arithmetical expression $P$ have the value $v$) with register values $v_1, \ldots, v_m$ if and only if 
the term $t_C$ (resp., $t_P$) defines the labelling that returns  $1$ (resp., returns $v$) provided that the paths corresponding to the registers store the nodes 
$v_1, \ldots, v_m$.

Before the construction, we remove all existential quantifiers from register constraints: for each existential quantification $\exists r . C$ we introduce a fresh path $p'$ and, then,
we substitute $r$ with a fresh register $r'$ that takes the value in the path $p'$. Basically, at every position the path $p'$ stores the value of the existential quantification.

Now, assume that $C$ and $P$ contain no subexpression of the form $\exists r . C$. We construct the terms $t_C, t_P$ by induction.
For clarity, we assume that all terms take all paths as arguments (therefore we do not have to specify path variables).
In the following, we denote by $y_r$ the path variable that stores the values of the register $r$.

\noindent \emph{Basis of induction}.
As the induction basis for register constraints we take $C :=  r = r  \mid E(r, e, r) \mid \query{Q}(r, \dots, r)$.
For $C := r = r'$, we define $t_C := y_{r} = y_{r'}$.
For $C := E(r, e, r')$, we define $t_C := \lambda_E(y_{r}, c_e, y_{r'})$, where $c_e$ is a constant corresponding to $e$.
Finally, for $C := \query{Q}(r, \dots, r)$, by the main induction assumption, there is an \opra{} query $Q'$, which is se-equivalent to $Q$.
Then, we define $t_C := [Q']$.

As the induction basis for arithmetical expressions we take $P := \lambda(r)$, which we transform to $t_P = \lambda_{i}(y_r)$, where 
 $\lambda_{i}$ is the labelling in the standard embedding that corresponds to the standard labelling of nodes $\lambda$ in the data graph.

\noindent \emph{Induction step}. For $C := C_1 \lor C_2$, we define $t_C := max(t_{C_1},t_{C_2})$.
For $C := \neg C_1$, we define $t_C := 1 - t_{C_1}$.
For $C: = f_1(P_1, \dots, P_k) \leq f_2(P_1', \dots, P_l')$, we define 
$t_C := f_1(t_{P_1}, \dots, t_{P_k}) \leq f_2(t_{P_1'}, \dots, t_{P_l'})$.

For $P := f[r : C]$, we define $t_P := f(\{\lambda_i(y_r) \mid t_C \})$, where 
$\lambda_{i}$ is the labelling in the standard embedding that corresponds to the standard labelling of nodes $\lambda$ in the data graph.
\end{proof}
\smallskip

The proof of (2) of Theorem 1 consists of the following two lemmata. We show that for every binary LARE query $Q(x,y)$, we can decide in $\PSpace$
whether $G \models \forall x \exists y Q(x,y)$ (Lemma~\ref{l:LARESandAREAo}). In contrast, there is an \opra{} query $Q^U$ for which deciding $G \models \forall x \exists y Q^U(x,y)$  is undecidable (Lemma~\ref{l:LARESandAREA}).
Therefore, there is no LARE query $Q'$, which is se-equivalent to $Q^U$.

\begin{lemma}
\label{l:LARESandAREAo}
For every \LARE{} query $Q(x,y)$, the problem, given a graph $G$, decide whether $\forall x \exists y Q(x,y)$ 
holds in $G$ is decidable in polynomial space. 
\end{lemma}

\begin{proof}
It has been shown in~\cite{Grabon+16} that for every graph $G$, 
we can compute in polynomial time an automaton $\aut_G$ 
recognizing pairs of paths $(\pi_1,\pi_2)$ from $G$ satisfying $Q(x, y)$ in $G$. 
We modify $\aut$ to nondeterministically guess a value for $y$ while processing the input. 
We do this by removing from the transitions the component corresponding to $y$. 
The resulting automaton $\aut'$ recognizes $P(x) = \exists y Q(x)$.  
Next, we check whether $\aut'$ accepts all paths in $G$, which can be done in polynomial space.
\end{proof}

\smallskip

\begin{lemma}
There exists a \opra{} query $Q^U(x,y)$ such that
the following problem is undecidable: given a graph $G$, decide whether $\forall x \exists y Q^U(x,y)$ holds in $G$.
\label{l:LARESandAREA}
\end{lemma}
We show how to construct a query and a graph that 
encode the universality problem for weighted automata with weights $-1,0,1$, 
which is undecidable~\cite{Almagor+2011}.
A weighted automaton is an automaton whose transitions are labelled by integers called weights. 
The value of a run is the sum of the weights of its transitions. 
The value of a word is the minimum over all values of accepting runs on this word.
The universality problem asks, given a weighted automaton, whether the value of every word is at most $0$. 

Given a weighted automaton $\aut$,
we define a graph $G_{\aut}$ whose nodes are transitions of $\aut$.
The graph $G_{\aut}$ has $4$ unary labellings $\lambda_c, \lambda_a, \lambda_s, \lambda_f$ and one binary labelling $\lambda_E$.
A node $u$ of $G_{\aut}$, corresponding to a transition $\alpha$, has the following labelling:
\begin{compactenum}
\item $\lambda_c(u)$ is the weight of the transition $\alpha$,
\item $\lambda_a(u)$ is the unique number corresponding to the letter labelling $\alpha$,
\item $\lambda_s(u)$  is $1$ if $\alpha$ starts from the initial state and $0$ otherwise, and 
\item $\lambda_f(u)$  is $1$ if $\alpha$ ends in a final state and $0$ otherwise.
\end{compactenum}
The labelling $\lambda_E(u_1, u_2)$ is $1$ if the destination state of (the transition corresponding to) $u_1$ is the same as the source of $u_2$, i.e., $u_1 = (q,a,q')$ and $u_2 = (q',b,q'')$ 
for some states $q, q', q''$ and some letters $a,b$.
Otherwise, $\lambda_E(u_1, u_2)$  is $0$.
Recall that we define a \emph{route} as a path $v_1, \ldots, v_k$ such that for all successive nodes $v_i, v_{i+1}$ we have $\lambda_E(v_i, v_{i+1}) = 1$.
Observe that the answer to the universality problem for $\aut$ is YES if and only if for every route $\pi_1$ from $G_{\aut}$
 there exists a route $\pi_2$ which corresponds to an accepting run on the word defined by $\pi_1$
of value at most $0$. Formally, $\pi_2$ has to satisfy the following:
(1)~\emph{($\pi_1$ and $\pi_2$ encode the same word)} the labels $\lambda_a$  along $\pi_2$ are the same as along $\pi_1$,
(2)~\emph{($\pi_2$ is a run)} $\pi_2$ starts in (the node corresponding to) a transition from some initial state of $\aut$ and 
finishes in a transition to some accepting state of $\aut$ (we can encode that with $s$ and $f$ components of nodes' labels), and
(3)~\emph{(the value of $\pi_2$ does not exceed $0$)} the sum of labels $\lambda_c$   along $\pi_2$ 
does not exceed $0$, i.e., the arithmetical constraint $\lambda_c[\pi_2] \leq 0$ holds.

Let $Q^U$ be a query encoding (1), (2) and (3). It follows that $\forall x \exists y Q^U(x,y)$ holds in 
$G_{\aut}$ if and only if the answer to the universality problem for $\aut$ is YES.
This concludes the proof of undecidability.

\begin{proof}(of (2) from Theorem 1)
Lemma~\ref{l:LARESandAREAo}  implies that for the query $Q^U$ from Lemma~\ref{l:LARESandAREA} there exists no equivalent \LARE{} query.
\end{proof}

\Paragraph{ECRPQ with linear constraints}
ECRPQ has been extended with linear constrains (ECRPQ+LC)~\cite{Barcelo+12}, expressing that a given vectors paths $\vec{\pi}$ satisfying a given ECRPQ query satisfies linear inequalities, which specify the multiplicity of edge labels in various components of $\vec{\pi}$. 
Language \opra{} subsumes LARE, which extends ECRPQ, and hence \opra{} subsumes ECRPQ.  
Linear constraints can be expressed by arithmetical constraints of \opra{} and hence \opra{} subsumes ECRPQ+LC. 
Moreover, linear constrains are unaffected by nodes' labels and hence ECRPQ+LC cannot express a \pra{} query saying 
``the sum of integer labels of nodes along path $p$ is positive''. 
Thus, we have the following. 

\begin{theorem}
(1)~\opra{} subsumes ECRPQ+LC.
(2)~There is an \opra{} query $Q$ with no ECRPQ+LC query $Q'$ se-equivalent to $Q$.
\end{theorem}
\begin{proof}
To prove the lemma, we need to formally define ECRPQs with linear constrains~\cite{Barcelo+12}.
A linear constraint is given by $h>0$, a $h \times (|\Sigma| \cdot n)$  matrix $A$ with integer coefficient and a vector $\vec{b} \in \Z^{h}$. 
A tuple of paths $\vec{\pi}$ satisfies this constraint if $A \vec{l} \leq \vec{b}$ holds for the vector 
$\vec{l} = (l_{1,1}, \ldots, l_{|\Sigma|,1}, l_{1,2}, \ldots, l_{|\Sigma|, n})$, where $l_{j,i}$ is the number of occurrences of the $j$-th edge label
on the $i$-th component of $\vec{\pi}$.
In ECRPQ+LC, we require a tuple of paths to satisfy both the ECRPQ part and linear constraints.

\noindent \textbf{(1):}
Language \opra{} can express all ECRPQ queries with linear constraints.
First, the query language \LARE{} can express ECRPQs~\cite{Grabon+16}. 
Second, similarly to the example {Processed labellings} from Section~\ref{s:opra}, we can define labelling $\lambda_{i,a}$, which is $1$
if in path $\pi_i$,  the current node (in the standard embedding) represents an edge labelled by $a$ (from the data graph). 
Thus, having an ECRPQ query with linear constraints $A \vec{l} \leq \vec{b}$, we compute:
(1)~a LARE{} query $\varphi_E$ corresponding to the ECRPQ part, 
(2)~a labellings $\lambda_{i,a}$ for every path $\pi_i$ from $\varphi_E$ and every edge  label $a$ from $\Sigma$, 
(3)~the arithmetical constraint $\varphi_A$ over $\lambda_{i,a}$ corresponding to $A \vec{l} \leq \vec{b}$. 
The conjunction of $\varphi_E \wedge \varphi_L \wedge \varphi_A$ is equivalent to 
 the given  ECRPQ query with linear constraints.

\noindent \textbf{(2):}
Linear constrains are unaffected by nodes' labels. The ECRPQ part of ECRPQ+LC does not refer to nodes labels as well. 
Therefore, ECRPQ+LC cannot express a \pra{} query saying ``the sum of integer labels of nodes along path $p$ is positive''. 
\end{proof}

 \section{The query answering problem}\label{sec:summing}
The query-answering problem asks, given an \opra{} $Q(\vec{x},\vec{\pi})$, a graph $G$, nodes $\vec{v}$ and paths $\vec{p}$ of $G$, whether $Q(\vec{v},\vec{p})$ holds in $G$.
 We are interested in the \emph{data complexity} of the problem, where the size of a query 
is treated as constant, and \emph{combined complexity}, where there is no such restriction.

To obtain the desired complexity results, we assume that the absolute values of the graph labels are polynomially bounded in the size of a graph. 
This allows us to compute arithmetical relations on these labels in logarithmic space. 
Without such a restriction, the data complexity of the query-answering problem we study is \NP{}-hard by a straightforward reduction from the knapsack problem. 

We state the complexity bounds as follows.

\begin{theorem}
The query answering problem 
for \opra{} queries 
with bounded number of auxiliary labellings is \PSpace{}-complete and its data complexity is \NL{}-complete.
\label{th:query-in-NL}
\end{theorem}

The emptiness problem (whether there exist nodes $\vec{v}$ and paths $\vec{p}$ such that a given \opra{} query $Q$ holds for $\vec{v},\vec{p}$ in a given graph $G$) has the same complexity; this follows from the fact that a query $Q(\vec{x}, \vec{\pi})$ is non-empty in $G$ iff $Q(\epsilon, \epsilon)$ (same query without free variables) holds in $G$. 

The lower bounds in Theorem~\ref{th:query-in-NL} follow from the \PSpace{}-hardness of ECRPQ~\cite{Barcelo+12}, as discussed in Section~\ref{sec:ep},
and for the \NL{}-hardness of the reachability problem.

Recall that an \opra{} query is of the form
\q{LET $O$ IN $Q'$}, where $Q'$ is a \pra{} query, $O$ is of the form $\definedBy{\lambda_1}{t_1}, \ldots, \definedBy{\lambda_n}{t_n}$ and $t_1, \ldots, t_n$ are terms. Also, by $|O|$ we denote the number of labellings defined in $O$.
 The upper bound in Theorem 3 follows directly from the following lemma.

\begin{lemma}
\label{l:induction}
For every fixed $s \geq  0$, we have:
\begin{enumerate}[(1)]
\item Given a graph $G$ and an \opra{} query $Q$  := \q{LET $O$ IN $Q'$} such that $|O| \leq s$, we can decide whether $Q$ holds in $G$ in non-deterministic polynomial space in $Q$ and non-deterministic logarithmic space in $G$.
\item Given a graph $G$ and an \opra{} query $Q$  := \q{LET $O$ IN $Q'$} such that $|O| \leq s$, we can compute $\min_{\lambda, \pi} Q(\vec{y}, \pi)$ (resp., $\max_{\lambda, \pi} Q(\vec{y}, \pi)$) non-deterministically in  
 polynomial space in $Q$ and logarithmic space in $G$. The computed value is either polynomial in $G$ and exponential in $Q$, or $-\infty$ (resp., $\infty$).
\end{enumerate}
\end{lemma}

We first prove the upper bounds for \pra{} (i.e., for $s=0$), and then extend the results to \opra{}.

\subsection{Language \pra{}} 
Assume a \pra{} query \q{$Q=$ MATCH NODES $\vec{x}$, PATHS $\vec{\pi}$ SUCH THAT $P$ WHERE $R$ HAVING $A$}.
We prove the results in two steps.  First, we construct a Turing machine of a special kind (later on called \qam{}) that represents graphs, called \emph{answer graphs}, with distinguished initial and final nodes, such that every path from an initial node to a final node in this graph is an encoding of paths that satisfy constraints $P$ and $R$ of $Q$ in graph $G$ (for some instantiation of variables $\vec{x}$). These graphs are augmented with the computed values of expressions that appear in arithmetical constraints $A$. Then, we prove that checking whether in an answer graph there is a path from an initial node to a final node that encodes a path in  $G$ satisfying $A$ can be done within desired complexity bounds.
Deriving values for $\vec{x}$ from computed paths is straightforward.

The first step is an adaptation of the technique commonly used in the field, e.g., in \cite{Barcelo+12,Grabon+16}. We encode vectors $(p_1,\dots, p_n)$ of paths of nodes from some $V$ as a single path $p_1 \otimes \ldots \otimes p_n$
over the product alphabet $V^n$ (shorter paths are padded with $\padding$).

\Paragraph{Answer graphs}
Consider a graph $G$ with nodes $V$, its paths $\vec{p}$ and an \opra{} query
\q{$Q=$ MATCH NODES $\vec{x}$, PATHS $\vec{\pi}$ SUCH THAT $P$ WHERE $R$ HAVING $\bigwedge_{i=1}^m A_i \leq c_i$} with $|\vec{p}| = |\vec{\pi}|$ and existentially quantified path variables $\vec{\pi}'$.
Let $k={|\vec{\pi}|+|\vec{\pi}'|}$.  
The \emph{answer graph} for $Q$ on $G$, $\vec{p}$ is a triple $(G', S, T)$, where $S, T\subseteq M_Q \times V^k$;

\begin{itemize}
\item $G'$ is a graph with nodes $M_Q \times V^k$, where $M_Q$ is a finite set computed from $Q$;
\item for each $i\leq m$ and a node $(s,v_1, \dots, v_k)\in M_Q \times V^k$, the labelling $\lambda_i((s,v_1, \dots, v_k))$ is defined as the value of the arithmetic constraint $A_i$ over single-node paths $v_1$, \dots, $v_{k}$;
and 
\item{}
 a path $q=q_1 \otimes \dots \otimes q_k$ from (a node in) $S$ to (a node in) $T$ is 
such that $\lambda_E(v, v')=1$ for all consecutive $v$, $v'$ of $q$ iff
 the paths $q_1, \dots, q_k$ satisfy $P \land R$ and $(q_1, \dots , q_{|\pi|}) = \vec{p}$.
\end{itemize}
Intuitively, the labelling $\lambda_E$ can be defined in such a way that along paths of $G'$  
the $V^k$-components of the nodes correctly encode paths of $Q$ satisfying the path constraints $P$ and
the $M_Q$ components store valid runs of automata corresponding to the regular constraints $R$. 

\Paragraph{\qam{}s} Answer graphs can be represented (on-the-fly) in logarithmic space.
A Query Applying Machine (\qam{}) is a non-deterministic Turing Machine which works in logarithmic space and only accepts inputs encoding tuples of the form $(G, t, w)$, where 
$G$ is a graph and $t$ is a symbol among ${ V, \lambda, S, T }$. 

For a graph $G$ and $k\geq 0$, a \qam{} $\M$ \emph{gives} a graph $G^\M_k = (V, \lambda_1, \dots, \lambda_k)$ and sets of nodes $S^\M_G$, $T^\M_G$ such that:
\begin{itemize}
\item $V$ consists of all the \emph{nodes} $v$ s.t. $\M$ accepts on $(G, V,  v)$, 
\item $\lambda_i$ is such that $\lambda_i(\vec{v}) = n$ iff $\M$ accepts on $(G, \lambda, (i, \vec{v}, n))$
\item $S_G^\M$ (resp., $T_G^\M$) consists of $v\in V$ such that $\M$ accepts on $(G, S, v)$ (resp., $(G, T, v)$).
\end{itemize}
For soundness, we require that for each $G$, $i$ and $\vec{v}$ there is exactly one $n$ such that $\M$ accepts on $(G, \lambda, (i, \vec{v}, n))$.

\begin{lemma}\label{l:pratoqam}
For a given query $Q$ and paths $\vec{p}$, we can construct in 
polynomial time a \qam{} $\M^Q$ such that for every graph $G$, machine $\M^Q$ gives an answer graph for $Q$ on $G, \vec{p}$.
\end{lemma}
\begin{proof}
\newcommand{\idxInp}{k_1}
\newcommand{\idxEx}{k_2}
Consider a graph $G$ with nodes $V$, its paths $\vec{p}$ and an \opra{} query
\q{$Q=$ MATCH NODES $\vec{x}$, PATHS $\vec{\pi}$ SUCH THAT $P_1, \ldots, P_l$ WHERE $R_1, \ldots, R_p$ HAVING $\bigwedge_{i=1}^m A_i \leq c_i$}
 with $|\vec{p}| = |\vec{\pi}|$ and existentially quantified path variables $\vec{\pi}'$.
Let $k={|\vec{\pi}|+|\vec{\pi}'|}$, i.e., we put $\idxInp = |\vec{\pi}|$ and $\idxEx = |\vec{\pi}'|$. 

First, let $N$ be the maximal length of paths in $\vec{p}$.
Second, for every regular constraint $R_i$ in $Q$, we build an NFA $\aut_i$ recognizing the language of the node constraints defined by $R_i$. 
We extend each $\aut_i$ with self-loops on all final states labelled by an auxiliary node condition $\bot$.
Intuitively, $\bot$ is true on paths that have terminated.
We define $M$ as a set of tuples $(s_1, \ldots, s_m, j)$ such that (1)~for $i \in \{1,\ldots, m\}$ we have
$s_i$ is a state of $\aut_i$, and (2)~$j \in \{1, \ldots, N, \infty\}$.
As defined in Section~\ref{sec:summing}, the set of nodes of $G'$ is $M \times V^k$ with additional constraint that a tuple 
$((s_1, \ldots, s_p, j), v_1, \ldots, v_k)$ is a node of $G'$ if for all $i \in \{1, \ldots, \idxInp\}$ we have $v_i = p_i[j]$, i.e., $v_i$ is the $j$-th node of the input path $p_i$ (or $\padding$ if $j > |p_i|$). 
The component $V^k$ represents paths in $Q$; since paths can have different length, we assume that the shorter paths are padded with the $\padding$ node that has special status. 
The auxiliary node condition $\bot$ holds if and only if all nodes are $\padding$, i.e., it allows us to combine regular conditions that work on paths of different length, i.e., 
if a regular condition represented by automaton $\aut_i$ holds on some  sequence of paths, we can append to all these paths $\padding$ nodes and $\aut_i$ still accepts.
Observe that given a graph $G$, we can recognize nodes of $G'$ in logarithmic space (note that $\vec{p}$ is fixed).

Next, we define the edge labelling $\lambda_E$ as follows:  Let $\vec{u}_1, \vec{u}_2$ be nodes of $G'$.
We have $\lambda_E(\vec{u}_1, \vec{u}_2) = 1$ if nodes $\vec{u}_1, \vec{u}_2$ are (1)~\emph{path consistent} and (2)~\emph{state consistent}, defined as follows. Otherwise, 
$\lambda_E(\vec{u}_1, \vec{u}_2) = 0$.
We now define path and state consistency. 
Let $\vec{u}_1 = (m, v_1, \ldots, v_k)$ with $m = (s_1, \ldots, s_p, j)$, and let
 $\vec{u}_2 = (m', v_1', \ldots, v_k')$ with $m' = (s_1', \ldots, s_p', j')$.
 Path consistency  ensures that (1)~the~first $\idxInp$ components of $V^k$ encode input paths $\vec{p}$, i.e.,
\begin{compactitem}
\item $j' = j+1$ if $j <N$ and $j' = \infty$ if $j \geq N$, and each $i \in \{1, \dots, \idxInp\}$ we have $v_i = p_i[j]$ or $v_i = \padding$ and $j > |p_i|$ (in particular if $j = \infty$).
\end{compactitem}
and (2a)~paths that end should satisfy path constraints and (2b)~paths that has terminated do not restart (all paths terminate with $\padding$), i.e.,
\begin{compactitem}
\item for each $i \in \{\idxInp+1, k\}$, we require two conditions (a)~if $v_i \neq \padding$ and $v_i' = \padding$ ($v_i$ is the last node of $\pi_i$), then for every 
path constraint $\piszczalka{x_s}{\pi_i}{x_t}$, we require $v_i = x_t$, and 
b)~if $v_i = \padding$, then $v_{i}' = \padding$ (path that terminated does not restart).
\end{compactitem}

State consistency ensures that the component $M$ stores valid runs of automata $\aut_1, \ldots, \aut_p$ corresponding to register constraints, i.e.,
\begin{compactitem}
\item for each $i \in \{1, \ldots, p\}$, the automaton $\aut_i$ has a transition $(s_i, a_i, s_i')$, where $a_i$ is a node constraint, and
this node constraint $a_i$ is satisfied over the nodes of $v_1, v_1', \ldots, v_k, v_k'$ (we assume that $a_i$ selects from the list of all paths only the relevant paths listed in the regular constraint). 
Recall that if $a_i = \bot$, then all nodes are required to be $\padding$. 
\end{compactitem}
It easy to check that given a graph $G$ and its two nodes $\vec{v}_1, \vec{v_2}$, we can decide in logarithmic space in $G$ whether 
$\lambda_E(\vec{v}_1, \vec{v}_2)$ is $0$ or $1$.

Then, we define labellings $\lambda_1, \ldots, \lambda_m$ for each arithmetical constraint, i.e.,
for $i \in \{1, \ldots, m\}$ we define $\lambda_i((m, v_1, \ldots, v_k))$ as $A_i$ computed on the subset of $v_1, \ldots, v_k$ selected by $A_i$.
Since $A_i$ is a linear combination and each labelling of $G$ is given in unary, all labellings of $G'$ can be computed in logarithmic space in $G$.

Finally, we define $S, T$ as follows. The set $S$ consists of nodes $(m, v_1, \ldots, v_k,1)$ such that 
\begin{compactenum}[(1)]
\item $m = (s_1, \ldots, s_m)$ and for $i \in \{1, \ldots, m\}$, we have $s_i$ is an initial state of $\aut_i$, 
\item for every path constraint $\piszczalka{x_s}{\pi}{x_t}$, we require $v_i = x_s$, and
\item for every $i \in \{1, \ldots, \idxInp\}$ we have $v_i = p_i[1]$.
\end{compactenum}
 The set $T$ consists of nodes $(m, v_1, \ldots, v_k)$ such that 
\begin{compactenum}[(1)]
\item $m = (s_1, \ldots, s_m,\infty)$ and for $i \in \{1, \ldots, m\}$, we have $s_i$ is a final state of $\aut_i$, 
\item for every $i \in \{1, \ldots, k\}$ we have $v_i = \padding$, i.e., all paths have terminated.
\end{compactenum}
\end{proof}

The second step amounts to the following lemma. 

\begin{lemma}\label{l:solvingqam}
For a graph $G$ and a \qam{} $\M^Q$, let $\Pi$ be the set of paths from $S_G^{\M^Q}$ to $T_G^{\M^Q}$ satisfying $\bigwedge_{i=1}^m \lambda_i[\pi] \leq c_i$ in $G_m^{\M^Q}$.
(1)~Checking emptiness of $\Pi$ can be done non-deterministically in polynomial space in $Q$ and logarithmic space in $G$.
(2)~Computing the minimum (resp., maximum) of the value $\lambda_j[\pi]$ over all paths in $\Pi$ can be done non-deterministically in polynomial space in $Q$ and logarithmic space in $G$.
The computed extremal value is either polynomial in $G$ and exponential in $Q$, or $-\infty$ (resp., $\infty$).
\end{lemma}
\begin{proof}
\newcommand{\norm}[1]{\|#1\|_1}
\newcommand{\cost}{\mathbf{c}}
A \emph{vector addition system with states}  (VASS) is an $\Z^d$-labelled graph $G$, i.e., $G = (V,E)$, where $V$ is a finite set and $E$ is a finite subset of $V \times \Z^d \times V$.
Depending on representation of labels $Z^d$, we distinguish unary and binary VASS. We define $\mu(s,\vec{v}, s')$ as $v$, the label of edge $(s,\vec{v}, s')$.
The \emph{$\Z$-reachability problem} for VASS, asks, given a VASS $G$ and its two \emph{configurations} $(s,\vec{u_1}), (t, \vec{u_2}) \in V \times Z^d$, whether there exists a path $\pi$ from 
$s$ to $t$ (of length $k$) such that  $\vec{u_1} + \sum_{i=0}^{k} \mu(\pi[k]) = \vec{u_2}$, i.e., $\vec{u_1}$ plus the sum of labels along $\pi$ equals $\vec{u_2}$. 

We discuss how to reduce the problem of existence of a path in $\Pi$ to {$\Z$-reachability problem} for VASS. We transform $G_m^{\M^Q}$ into a VASS $G' = (V,E)$ over the set of nodes of $G_m^{\M^Q}$
with two additional nodes $s,t$.  
We put an edge between two nodes connected node $q_1, q_2 \in V$ labelled by the label of the source node $\vec{v}$, i.e., for all $q_1, q_2 \in V$ we have $(q_1, \vec{v},q_2)$ if $\lambda_E(q_1, q_2) = 1$ and 
$\vec{v} = (\lambda_1(q),\ldots, \lambda_m(q))$. Moreover, 
we define $s$ as the source and $t$ and the sink, i.e., 
(1)~for every $q \in s_G^{\M^Q}$ we put an edge $(s,\vec{v}, q)$, where $\vec{v} = (c_1,\ldots, c_m)$ (constants from the definition of $\Pi$),
and (2)~for every $q \in T_G^{\M^Q}$ we put an edge $(q,\vec{v}, t)$, where 
$\vec{v} = (\lambda_1(q),\ldots, \lambda_m(q))$.
Finally, we allow the labels to be increased in $t$, i.e., 
for every $i \in \{1, \ldots, m\}$, we put $(t, \mathbbm{1}_i, t)$, where $\mathbbm{1}_i \in Z^d$ has $1$ at the component $i$, and $0$ at all other components.

Observe that paths from $\Pi$ correspond to paths in VASS $G'$ from $(s,\vec{0})$ to $(t,\vec{0})$. The $\Z$-reachability problem for unary VASS of the fixed dimension (which is $m$ in the reduction) is in \NL~\cite[Therem 19]{Blondin+15}. In the proof, it has been shown that if there exists a path between given two configurations, then there also exists a path, which is 
(a)~polynomially bounded in the VASS, provided that weights are given  in unary and the dimension is fixed, and
(b)~exponentially bounded in the VASS without restrictions. 

Now, we are ready to show (1). 
If $Q$ is fixed, the dimension of $G'$, number $m$ is fixed as well, and all labelling values are given in the unary. 
Therefore, if $\Pi$ is non-empty, then it contains a path of polynomial size in $|G|$. Existence of such a path can be verified in non-deterministic 
logarithmic space using  the \qam{} $\M^Q$  as an oracle to query for nodes, edges and labelling of $G_m^{\M^Q}$.
The \qam{} $\M^Q$ requires logarithmic space in $|G|$.
Therefore checking emptiness of $\Pi$ can be done non-deterministically in logarithmic space in $|G|$.

If $\Pi$ is non-empty, then it contains a path of exponential size in $|Q|$. Existence of such a path can be verified in non-deterministic 
polynomial space using  the \qam{} $\M^Q$  as an oracle to query for nodes, edges and labelling of $G_m^{\M^Q}$.
Therefore checking emptiness of $\Pi$ can be done non-deterministically in polynomial space in $Q$.

To show (2), we need to analyze the proof of~\cite[Theorem 19]{Blondin+15}. It has been shown that there exist a finite set $S$ of path schemes of the form 
$\alpha_0 \beta_1^* \ldots \beta_k^* \alpha_k$ such that (1)~each path scheme in $S$ has 
polynomially bounded length in the size of VASS,
for all configurations $(s,\vec{v}_1)$, $(t,\vec{v}_2)$, if there is a path from $(s,\vec{v}_1)$ to $(t,\vec{v}_2)$, then there is a path that matches some path scheme from $S$.
Next, it has been shown that for every  path scheme $\alpha_0 \beta_1^* \ldots \beta_k^* \alpha_k = \rho \in S$, there is a system of linear Diopahntine equations $\mathcal{E}_{\rho}$ such that  
$\mathcal{E}_{\rho}$ has a solution $x_1, \ldots, x_k$ if and only if
$\alpha_0 \beta_1^{x_1} \ldots \beta_k^{x_k} \alpha_k$ is a path from $(s,\vec{v}_1)$ to $(t,\vec{v}_2)$.
Next, for each system of linear Diophantine equations $\mathcal{E}_{\rho}$, its form implies that all its solutions are represented by $B_{\rho} + cone(P_{\rho})$, where 
$B_{\rho},P_{\rho}$ are sets of vectors whose coefficients are 
(a)~exponentially bounded in the dimension, 
(b)~polynomially bounded in the size of VASS (with the dimension fixed).
Finally,  $cone(P_{\rho})$ linear combinations of vectors from $P$ with non-negative integer coefficients.
It follows that if $\Pi$ is non-empty, one of the following holds:
\begin{enumerate}[(1)]
\item For some path scheme $\rho = \alpha_0 \beta_1^* \ldots \beta_k^* \alpha_k$, sets $B_{\rho},P_{\rho}$ are non-empty ( $\mathcal{E}_{\rho}$  has infinitely many solutions), and for some
$\vec{u} = (u_1, \ldots, u_k)$, the sum of the value $\lambda_j[\pi]$ over paths $\beta_1^{u_1}, \ldots, \beta_k^{u_k}$ is negative, and hence the minimum is $-\infty$,
\item Otherwise, 
the minimum exists and it is realized by some path $\pi$ matching some path scheme of the form $\rho = \alpha_0 \beta_1^* \ldots \beta_k^* \alpha_k$ such that 
$\pi = \alpha_0 \beta_1^{x_1} \ldots \beta_k^{x_k} \alpha_k$, where $(x_1, \ldots, x_k) \in B_{\rho}$. Observe that it does not pay off to incorporate vectors from $P_{\rho}$
as they cannot decrease the value of $\lambda_j[\pi]$. Finally, observe that the size of such a path $\pi$ is polynomial in the size of VASS if the dimension of the VASS is fixed and
it is exponential in the dimension.
\end{enumerate}

From (1) and (2), we derive bounds $b_1(G,Q) < b_2(G,Q)$, which are polynomial in $G$ and exponential in $Q$ such that
if  the minimum of $\lambda_j[\pi]$ over path in $\Pi$ exists, then it is realized by some path of length bounded by $b_1(G,Q)$.
However, if there is a path $\pi \in \Pi$ of length between $b_1(G,Q)$ and $b_2(G,Q)$, with $\lambda_j[\pi]$ lower than the value of any path shorter than 
$b_1(G,Q)$,  then the minimum is infinite. 
Since NL and PSPACE and closed under complement and bounded alternation (only two conditions to be checked), both conditions
can be checked in non-deterministically in polynomial space in $Q$ and logarithmic space in $G$.
 The case of the maximum is symmetric.
\end{proof}

\subsection{Language \opra{}}
Assume $O=\, \definedBy{\lambda_1}{t_1}$, \dots, $\definedBy{\lambda_s}{t_s}$. We show by induction on $s$ that the values of the labellings of a graph $G[O]$ can be non-deterministically computed in space polynomial in $O$.

\begin{lemma}\label{l:step}
Let $s$ be fixed. For a graph $G$ and $O=\, \definedBy{\lambda_1}{t_1}$, \dots, $\definedBy{\lambda_s}{t_s}$, 
the value of each labelling of $G[O]$ can  be non-deterministically computed in polynomial space in $O$ and
logarithmic space in $G$.
\end{lemma}
\begin{proof}
The proof is by induction on $s$.
The basis of induction, $s = 0$, is trivial.
Assume that for $s$ the lemma statement and Lemma~\ref{l:induction} hold.  
We show that it holds for $s+1$.
Consider graph $G$,  $O=\, \definedBy{\lambda_1}{t_1}$, \dots, $\definedBy{\lambda_s}{t_s}$, and 
$O' =\, O$, $\definedBy{\lambda_{s+1}}{t_{s+1}}$.

We show that the value of $t_{s+1}$ can  be non-deterministically computed in polynomial space in $O$ and
logarithmic space in $G$. We start the computation form the bottom, the leaves, and show that 
the values of leaves can  be non-deterministically computed in polynomial space in $O$ and
logarithmic space in $G$. Indeed, leaves are of one of the following forms: 
$c \mid \lambda(\vec{y}) \mid [Q(\vec{y})] \mid \min_{\lambda, \pi} Q(\vec{y}, \pi) \mid \max_{\lambda, \pi} Q(\vec{y}, \pi) \mid y = y$.
The cases of $c$ and $y = y$ are trivial. 
For leaves of the form $\lambda(\vec{y})$, we use the inductive assumption of this lemma. 
Finally, for leaves of the form $[Q(\vec{y})]$ and $\min_{\lambda, \pi} Q(\vec{y}, \pi)$, it follows from Lemma~\ref{l:induction} applied inductively for $s$.

The internal nodes of $t_{s+1}$ are of the form 
$ f(t_1(\vec{y}), \dots, t_k(\vec{y}) ) \mid f'(\{ t(x) \colon t(x,\vec{y})\})$.
Having the values of subterms $t_1(\vec{y}), \dots, t_k(\vec{y})$, the value $f(t_1(\vec{y}), \dots, t_k(\vec{y}) )$ 
can be computed in logarithmic in length of $\vec{x}$ and values in $\vec{x}$, i.e., we require space $max(\log(|t_1(\vec{y})|), \ldots, \log(|t_k(\vec{y})|)) + \log(k) + C$, where $C$ is a constant.
Similarly, to compute  $f'(\{ t(x) \colon t(x,\vec{y})\})$, we require space $max_x(\log(|t_1(x,\vec{y})|) + \log(|G|) + C$, where $C$ is a constant. 
It follows that to compute the value of $t_s$, we require space $|t_s| (\log|G| + C) \cdot M$, where $C$ is the maximal constant taken over all $f \in \F$ (which is fixed), and
$M$ is the maximum over space requirements of the leaves, which is logarithmic in $G$ and polynomial in $t_s$.
\end{proof}

Finally, we are ready to prove Lemma~\ref{l:induction}.

\begin{proof}(of Lemma~\ref{l:induction})
The proof is by induction on $s$.
 Lemma \ref{l:pratoqam} and Lemma \ref{l:solvingqam}' imply the basis of induction.
Next, assume that this lemma holds for $s$. 
Consider a query \q{LET $O$ IN $Q'$}, with $|O| = s+1$ and a graph $G$.  We first build a \qam{} $\M^{Q'}$, as in Lemma~\ref{l:pratoqam}. $M^{Q'}$ may refer to labellings from $O$, not defined in $G$. 
We change it so that whenever it wants to access a value of one of labellings defined in $O$, it instead runs a procedure guaranteed by Lemma~\ref{l:step}. 
Lemma~\ref{l:step} holds because this lemma holds for $s$.
Finally, we use Lemma~\ref{l:solvingqam} to determine the result.
\end{proof}

 \section{Conclusions}
\sloppy
We defined a new query language for graph databases, \opra{} and demonstrated its expressive power in two ways.
We presented examples of natural properties and \opra{} queries expressing them in an organized, modular way. 
We showed that \opra{} strictly subsumes query languages ECRPQ+LC and LARE.
Despite additional expression power, the complexity of the query-answering problem for  \opra{} matches the complexity for ECRPQ+LC and LARE. 

\fussy \bibliographystyle{plain}
\bibliography{bib}

\end{document}